\documentclass[11pt]{article}
\usepackage[a4paper,margin=1in]{geometry}
\usepackage{amsmath,amssymb,amsthm,bm}
\usepackage{graphicx}
\usepackage{hyperref}
\usepackage{amsmath,amsfonts}
\usepackage{algorithmic}
\usepackage{array}
\usepackage[caption=false,font=normalsize,labelfont=sf,textfont=sf]{subfig}
\usepackage{textcomp}
\usepackage{stfloats}
\usepackage{url}
\usepackage{indentfirst}

\usepackage{verbatim}
\usepackage{graphicx}
\def\BibTeX{{\rm B\kern-.05em{\sc i\kern-.025em b}\kern-.08em
		T\kern-.1667em\lower.7ex\hbox{E}\kern-.125emX}}
\usepackage{balance}
\usepackage{amssymb}
\usepackage{algorithm}
\usepackage{cite}
\usepackage{bm}
\usepackage{cuted}
\usepackage{placeins}

\usepackage[utf8]{inputenc}
\usepackage{lineno} 
\usepackage{cases}
\usepackage{float}  
\usepackage{textcomp}
\usepackage{bbold}
\usepackage{xcolor}
\usepackage{siunitx}

\usepackage{caption}
\usepackage{amsthm}
\newtheorem{theorem}{Theorem}

\title{Optimality Analysis of RSMA Degenerating to SDMA Under Imperfect SIC}
\author{
	Xuejun Cheng$^{1}$\protect\footnotemark[2],
	Qian Zhang$^{1}$\protect\footnotemark[2],
	Yunnuo Xu$^{1}$,
	Zheng Dong$^{1}$,
	Ju Liu$^{1}$,
	Bruno Clerckx$^{2}$\\[2pt]
	\small $^{1}$School of Information Science and Engineering, Shandong University, Qingdao 266237, China\\
	\small $^{2}$Department of Electrical and Electronic Engineering, Imperial College London, London SW7 2AZ, U.K.\\
}
\date{}

\begin{document}
	\maketitle
	\begingroup
	\renewcommand{\thefootnote}{\fnsymbol{footnote}}
	\footnotetext[2]{Equal contribution.}
	\endgroup
	\setcounter{footnote}{0}

	\begin{abstract}
		This document serves as supplementary material for our journal submission, providing detailed mathematical proofs and derivations that support the results presented in the main manuscript.
		Specifically, under a modeling framework that jointly considers transceiver hardware impairments and imperfect successive interference cancellation (SIC), we systematically derive and prove from an optimality perspective that: when the residual interference coefficient approaches $1$ (i.e., SIC becomes severely ineffective), there exists an optimal solution such that the common stream beamformer satisfies $\bm w_c^\star=\bm 0$, and hence the optimal rate-splitting multiple access (RSMA) transmission structure degenerates into space division multiple access (SDMA). This conclusion provides a verifiable theoretical justification for the convergence phenomenon observed in simulations, namely that ``the RSMA performance gradually approaches that of SDMA as SIC degrades'', and can also serve as a reference for multiple-access selection and system design in SIC-limited scenarios.
	\end{abstract}
	
	\section{Introduction}
	
	Rate-splitting multiple access (RSMA) provides flexible interference management for MIMO systems by splitting each user’s message into common and private parts and transmitting them jointly. The basic idea is that the common stream can be decoded by multiple users and removed at the receiver, while the private stream is decoded only by its intended user. In this way, RSMA establishes a unified framework between space division multiple access (SDMA) and non-orthogonal multiple access (NOMA), and can adaptively adjust the transmission strategy according to channel conditions and interference structures.
	
	However, the system performance of RSMA is closely related to the effectiveness of successive interference cancellation (SIC) at the receiver. In practical systems, due to transceiver hardware non-idealities, channel estimation errors, and finite-precision digital processing, SIC is often far from ideal and inevitably results in residual interference. This residual component continuously enters the private stream decoding process as additional interference, thereby weakening the interference management gain brought by common/private splitting and degrading system performance. Although existing studies have observed at the level of intuitive analysis and numerical simulations that the RSMA performance will gradually approach SDMA as SIC deteriorates, a rigorous optimality characterization of this convergence phenomenon is still relatively lacking.
	
	To this end, under a modeling framework that jointly considers transceiver hardware impairments and imperfect SIC, we provide a rigorous proof from an optimization optimality perspective: when the residual interference coefficient approaches $1$ (i.e., SIC is severely ineffective), there exists an optimal solution such that the common stream beamformer satisfies $\bm w_c^\star=\bm 0$, and thus the optimal RSMA transmission structure is mathematically equivalent to SDMA. This result provides theoretical support for understanding the performance boundary and structural degeneration mechanism of RSMA, and offers guidance for multiple-access selection and robust design under SIC-limited conditions.
	
	\section{System Model}
	
	We consider a downlink RIS-assisted RSMA system, where multiple practical non-ideal factors affecting system performance are included in the model, including transceiver hardware impairments (HWI) and imperfect successive interference cancellation (SIC). In this system, a base station (BS) equipped with \(M\) transmit antennas serves \(K\) single-antenna users with the aid of an RIS composed of \(N\) reflecting elements. The direct BS--user links are assumed to be blocked, and communication mainly relies on the ``BS--RIS--user'' reflective link. Let \(\mathbf{G}\in\mathbb{C}^{N\times M}\) denote the channel matrix from the BS to the RIS, and \(\bm f_k\in\mathbb{C}^{N\times 1}\) denote the channel vector from the RIS to user \(k\). The RIS phase-shift matrix is denoted by \(\bm \Theta=\mathrm{diag}\!\left(\bm \phi^{\mathrm{H}}\right)\in\mathbb{C}^{N\times N}\), where \(\bm \phi\) is the RIS reflection-coefficient vector.
	
	\subsection{Transmit Signal Model}
	
	The BS adopts the RSMA transmission strategy, and the transmitted signal is expressed as
	\begin{equation}\label{x}
		\bm x = \bm w_c s_c + \sum_{k=1}^{K}\bm w_k s_k + \bm \eta_t,
	\end{equation}
	where $s_c$ denotes the common data stream, $s_k$ denotes the private data stream for user $k$, and $\bm w_c$ and $\bm w_k$ are the corresponding beamforming vectors.
	
	The distortion noise caused by generalized transmitter hardware impairments (HWI) is modeled as
	\begin{equation}
		\bm \eta_t \sim \mathcal{CN}\!\left(\bm 0,\; m_t\, \widetilde{\mathrm{diag}}\!\left(\bm w_c \bm w_c^{\mathrm{H}}+\sum_{k=1}^{K}\bm w_k \bm w_k^{\mathrm{H}}\right)\right),
	\end{equation}
	which has been validated by both theoretical analysis and field measurements~\cite{schenk2008rf,bjornson2014massive,zhang2023robust}. Here, $\widetilde{\mathrm{diag}}(\cdot)$ denotes the operator that extracts the diagonal elements of the input matrix and forms a diagonal matrix; $m_t \in (0,1)$ denotes the ratio of transmit distortion noise power to transmit signal power.
	
	\subsection{Receive Signal Model}
	
	The received signal at user $k$ is
	\begin{align}\label{y_K}
		y_k
		&= \tilde{y}_k + \eta_{r,k}
		= \bm h_k^\mathrm{H} \bm x + \eta_{r,k} + n_k \notag\\
		&=  \underbrace{\bm h_k^\mathrm{H} \bm w_c s_c}_{\text{common part}}
		+ \underbrace{\bm h_k^\mathrm{H} \bm w_k s_k}_{\text{private part}}
		+ \underbrace{\sum_{i=1, i \neq k}^{K} \bm h_k^\mathrm{H} \bm w_i s_i}_{\text{interference by other users}}
		+ \underbrace{\bm h_k^\mathrm{H}\bm \eta_t + \eta_{r,k} + n_k}_{\text{noise}},
	\end{align}
	\noindent where $\bm h_k^\mathrm{H} = \bm f_k^\mathrm{H} \bm\Theta \mathbf{G}$, and $ n_k \sim \mathcal{C} \mathcal{N}\left(0,\sigma_k^2\right)$ denotes the additive white Gaussian noise (AWGN) at user-$k$.
	Here, $\tilde{y}_k$ denotes the undistorted received signal before the receiver hardware distortion term $\eta_{r,k}$ is added.
	The receiver distortion noise is modeled as $\eta_{r,k} \sim \mathcal{C} \mathcal{N}\left(0,m_r\mathbb{E}[|\tilde{y}_k|^2]\right)$ caused by receiver HWI~\cite{schenk2008rf,bjornson2014massive,zhang2023robust}, where $m_r \in (0,1)$ is the ratio of distorted noise power to undistorted received signal power.
	
	\section{SINR Expressions Under Imperfect SIC}
	
	At the user side, each user initially decodes the common stream $s_c$ by treating all private streams as interference. Subsequently, user-$k$ decodes its private stream after removing $s_c$ via SIC.
	However, unlike most existing designs that assume perfect cancellation of correctly decoded signals, practical SIC is imperfect and leaves non-negligible residual self-interference, which degrades the private stream decoding performance.
	Therefore, the signal-to-interference-plus-noise ratio (SINR) of user-$k$ under transceiver HWI and imperfect SIC is given as
	\vspace{-3pt}
	\begin{equation*}
		\gamma_{c,k} = \frac{|\bm h_k^\mathrm{H} \bm w_c|^2}{\sum_{i=1}^{K}{|\bm h_k^\mathrm{H} \bm w_i|^2}+\Phi_{c,k}},
		\qquad
		\gamma_{p,k} = \frac{|\bm h_k^\mathrm{H} \bm w_k|^2}{\sum_{i=1,i\neq k}^{K}{|\bm h_k^\mathrm{H} \bm w_i|^2}+\Phi_{p,k}},
	\end{equation*}
	\vspace{-3pt}
	
	\noindent where $
	\Phi_{c,k} = \bm h_k^\mathrm{H}\!\left[m_r \mathbf{A}+ m_t(1+m_r)\widetilde{\text{diag}}(\mathbf{A})\right]\!\bm h_k +\sigma^2$,
	$\Phi_{p,k} = \delta_{\text{SIC}}^2|\bm h_k^\mathrm{H}\bm w_c|^2+\Phi_{c,k}$,
	$\mathbf{A} = \bm w_c \bm w_c^\mathrm{H}+\sum_{k=1}^{K}{\bm w_k \bm w_k^\mathrm{H}}$,
	$\sigma^2 = (1 + m_r)\sigma_k^2$,
	and $0 \leq \delta_{\text{SIC}}\leq 1$ denotes the coefficient of SIC imperfection~\cite{10517628}.
	The detailed derivation of $\Phi_{c,k}$ and $\Phi_{p,k}$ is provided in Appendix A.
	In practical deployments, the BS typically cannot obtain the exact value of $\delta_{\text{SIC}}$. To guarantee that the SINR satisfies the requirements for the allocated data rate, this parameter must be configured appropriately. The typical range for $\delta_{\text{SIC}}$ is set between 4\% and 10\%~\cite{9406993}.
	
	The achievable rates of user-$k$ are given as
	$R_{k} = R_{c,k}+R_{p,k}=\log_2\left(1+\gamma_{c,k}\right) + \log_2\left(1+\gamma_{p,k}\right)$.
	To ensure successful decoding of the common stream $s_c$ by all users, its achievable rate is constrained as
	$R_{c} = \min{\{R_{c,1},\dots, R_{c, K}\}}$.
	Then, the overall system sum rate is
	$R_{\text{total}} = R_c +\sum_{k=1}^{K} R_{p,k}$.
	
	\section{Limit Analysis and Core Proposition}
	
	\begin{theorem}[Theorem 1]
		Consider a $K$-user downlink RSMA system. Under transceiver hardware impairments and imperfect SIC,
		the SINR of the private stream for user $k$ can be expressed as
		\begin{equation}
			\gamma_{p,k}
			=
			\frac{\big|\bm h_k^\mathrm{H}\bm w_k\big|^{2}}
			{\sum\limits_{\substack{i=1, i\neq k}}^{K}\big|\bm h_k^\mathrm{H}\bm w_i\big|^{2}
				+\delta_{\mathrm{SIC}}^{2}\big|\bm h_k^\mathrm{H}\bm w_c\big|^{2}
				+\Phi_{c,k}},
		\end{equation}
		where
		\begin{align}
			\Phi_{c,k}
			&= \bm h_k^\mathrm{H}\!\left[ m_r\bm A + m_t(1+m_r)\,\widetilde{\mathrm{diag}}(\bm A)\right]\!\bm h_k + (1+m_r)\sigma_k^{2}, \\
			\bm A
			&= \bm w_c\bm w_c^\mathrm{H} + \sum_{k=1}^{K}\bm w_k\bm w_k^\mathrm{H}.
		\end{align}
		When $\delta_{\mathrm{SIC}}\to 1$, for any system utility function that is monotonically nondecreasing with respect to the users' SINRs (which can characterize common performance metrics such as weighted sum rate, proportional fairness, and max--min fairness), the optimal solution of the RSMA system must satisfy
		\begin{equation}
			\bm w_c^{\ast}=\bm 0,
		\end{equation}
		i.e., the optimal strategy allocates no power to the common stream, and RSMA adaptively degenerates into SDMA.
	\end{theorem}
	
	\begin{proof}
		Under the total transmit power constraint, the system optimization problem can be formulated as
		\begin{align}
			\max_{\bm w_c,\{\bm w_k\}}\quad
			& U\!\left(\gamma_{p,1}^{\mathrm{RSMA}},\ldots,\gamma_{p,K}^{\mathrm{RSMA}}\right) \\
			\text{s.t.}\quad
			& \|\bm w_c\|^{2}+\sum_{k=1}^{K}\|\bm w_k\|^{2}\le P_{\max},
		\end{align}
		where $U(\cdot)$ denotes a system utility function that is monotonically nondecreasing with respect to the users' SINRs.
		
		We adopt a proof by contradiction. Assume that when $\delta_{\mathrm{SIC}}\to 1$, the optimal solution
		$\{\bm w_c^{\star},\bm w_1^{\star},\ldots,\bm w_K^{\star}\}$ satisfies $\bm w_c^{\star}\neq \bm 0$.
		We then construct a new set of beamforming vectors as follows:
		\begin{equation}
			\widetilde{\bm w}_c=\bm 0,\qquad \widetilde{\bm w}_k=\bm w_k^{\star},\ \forall k\in\{1,\ldots,K\}.
		\end{equation}
		It can be readily verified that this construction still satisfies the power constraint, since
		\begin{equation}
			\|\widetilde{\bm w}_c\|^{2}+\sum_{k=1}^{K}\|\widetilde{\bm w}_k\|^{2}
			=\sum_{k=1}^{K}\|\bm w_k^{\star}\|^{2}
			\le \|\bm w_c^{\star}\|^{2}+\sum_{k=1}^{K}\|\bm w_k^{\star}\|^{2}
			\le P_{\max}.
		\end{equation}
		Substituting the new beamforming vectors into the expression of the private stream SINR yields
		\begin{equation}
			\widetilde{\gamma}_{p,k}
			=
			\frac{\big|\bm h_k^\mathrm{H}\widetilde{\bm w}_k\big|^{2}}
			{\sum\limits_{\substack{i=1, i\neq k}}^{K}\big|\bm h_k^\mathrm{H}\widetilde{\bm w}_i\big|^{2}+\widetilde{\Phi}_{c,k}}
			=
			\frac{\big|\bm h_k^\mathrm{H}\bm w_k^{\star}\big|^{2}}
			{\sum\limits_{\substack{i=1, i\neq k}}^{K}\big|\bm h_k^\mathrm{H}\bm w_i^{\star}\big|^{2}+\widetilde{\Phi}_{c,k}},
		\end{equation}
		where $\widetilde{\Phi}_{c,k}$ is obtained by substituting
		\begin{equation}
			\widetilde{\bm A}=\sum_{k=1}^{K}\bm w_k^{\star}\bm w_k^{\star,H}
		\end{equation}
		into $\Phi_{c,k}$. Since $\bm w_c^{\star}\bm w_c^{\star,H}\succeq \bm 0$, we have
		\begin{equation}
			\bm A^{\star}
			=
			\bm w_c^{\star}\bm w_c^{\star,H}+\sum_{k=1}^{K}\bm w_k^{\star}\bm w_k^{\star,H}
			\succeq
			\sum_{k=1}^{K}\bm w_k^{\star}\bm w_k^{\star,H}
			=
			\widetilde{\bm A}.
		\end{equation}
		Moreover, $\widetilde{\mathrm{diag}}(\cdot)$ preserves the elementwise order on diagonal elements, which implies
		\begin{equation}
			\widetilde{\Phi}_{c,k}\le \Phi_{c,k}^{\star}, \quad \forall k.
		\end{equation}
		By comparing the SINRs under the two schemes, we obtain, for any user $k$,
		\begin{align}
			\widetilde{\gamma}_{p,k}-\gamma_{p,k}^{\star}
			&=
			\frac{\big|\bm h_k^\mathrm{H}\bm w_k^{\star}\big|^{2}}
			{\sum\limits_{\substack{i=1, i\neq k}}^{K}\big|\bm h_k^\mathrm{H}\bm w_i^{\star}\big|^{2}+\widetilde{\Phi}_{c,k}}
			-
			\frac{\big|\bm h_k^\mathrm{H}\bm w_k^{\star}\big|^{2}}
			{\sum\limits_{\substack{i=1, i\neq k}}^{K}\big|\bm h_k^\mathrm{H}\bm w_i^{\star}\big|^{2}+\big|\bm h_k^\mathrm{H}\bm w_c^{\star}\big|^{2}+\Phi_{c,k}^{\star}}.
		\end{align}
		Since $\bm w_c^{\star}\neq \bm 0$, except for the degenerate case where $\bm w_c^{\star}$ lies in the intersection of the nullspaces of $\{\bm h_k\}$, there exists at least one user $k$ such that $\big|\bm h_k^\mathrm{H}\bm w_c^{\star}\big|^{2}>0$, which leads to $\widetilde{\gamma}_{p,k}>\gamma_{p,k}^{\star}$.
		Because the utility function $U(\cdot)$ is monotonically nondecreasing with respect to the users' SINRs, it follows that
		\begin{equation}
			U\!\left(\widetilde{\gamma}_{p,1},\ldots,\widetilde{\gamma}_{p,K}\right)
			>
			U\!\left(\gamma_{p,1}^{\star},\ldots,\gamma_{p,K}^{\star}\right),
		\end{equation}
		which contradicts the optimality of $\{\bm w_c^{\star},\bm w_1^{\star},\ldots,\bm w_K^{\star}\}$.
		Therefore, in the limiting case $\delta_{\mathrm{SIC}}\to 1$, any solution with $\bm w_c\neq \bm 0$ cannot be optimal,
		and the optimal common stream beamformer must satisfy
		\begin{equation}
			\bm w_c^{\ast}=\bm 0.
		\end{equation}
		Substituting $\bm w_c^{\ast}=\bm 0$ into the private stream SINR expression yields
		\begin{equation}
			\gamma_{p,k}
			=
			\frac{\big|\bm h_k^\mathrm{H}\bm w_k\big|^{2}}
			{\sum\limits_{\substack{i=1, i\neq k}}^{K}\big|\bm h_k^\mathrm{H}\bm w_i\big|^{2}+\Phi_{c,k}},
		\end{equation}
		which is fully consistent with the conventional space division multiple access (SDMA) structure.
		
		This completes the proof. 
		\end{proof}
		In summary, under imperfect successive interference cancellation, as $\delta_{\mathrm{SIC}}$ increases gradually from $0$
		and approaches $1$, the optimal transmission structure of RSMA will adaptively degenerate into SDMA, i.e.,
		\begin{equation}
			\lim_{\delta_{\mathrm{SIC}}\to 1}\mathrm{RSMA}=\mathrm{SDMA}.
		\end{equation}

	\section{Conclusion}
	From an optimization optimality perspective, this document rigorously proved that under imperfect SIC, when SIC became fully ineffective, there existed an optimal solution of the RSMA system for which the common stream beamformer degenerated to the zero vector, and thus the optimal transmission structure was mathematically equivalent to SDMA. This result provided theoretical support for understanding the performance boundary and structural degeneration mechanism of RSMA, and served as a reference for multiple-access selection and system design in SIC-limited scenarios.

	\bibliographystyle{IEEEtran}
	\bibliography{myref}
\end{document}